\documentclass[3p]{elsarticle}
\usepackage{amsmath}
\usepackage{amssymb}
\usepackage{graphicx}
\usepackage{tikz}
\usepackage{tabularx}
\usepackage{ALgo}
\usepackage{todonotes}
\usepackage{lineno}
\usepackage{hyperref}

\def\pp{\mathinner{\ldotp\ldotp}}
\def\sa#1{\mbox{\tt #1}}

\def\ptr{\textit{ptr}}
\def\Nil{\textit{Nil}}
 \def\PV{\mathcal{P}}
 
 \def\PVW{\mathcal{P}_{\s{w}}}
 \def\QV{\mathcal{Q}} 
\def\pp{\ldotp\ldotp}
\renewcommand{\epsilon}{\varepsilon}

\def\s#1{\mbox{$#1$}}

\newtheorem{theorem}{Theorem}
\newdefinition{definition}[theorem]{Definition}
\newtheorem{lemma}[theorem]{Lemma}
\newtheorem{corollary}[theorem]{Corollary}
\newtheorem{observation}[theorem]{Observation}

\newproof{proof}{Proof}

\begin{document}

\begin{frontmatter}

\title{Fast Computation of Abelian Runs}

\author[palermo]{Gabriele Fici\fnref{Gabriele}}
\ead{Gabriele.Fici@unipa.it}

\author[warsaw]{Tomasz Kociumaka\fnref{Iuventus}}
\ead{kociumaka@mimuw.edu.pl}

\author[rouen]{Thierry Lecroq}
\ead{Thierry.Lecroq@univ-rouen.fr}

\author[rouen]{Arnaud Lefebvre}
\ead{Arnaud.Lefebvre@univ-rouen.fr}

\author[rouen]{\'Elise Prieur-Gaston}
\ead{Elise.Prieur@univ-rouen.fr}

\address[palermo]{Dipartimento di Matematica e Informatica, Universit\`a di Palermo, Italy}

\address[warsaw]{Faculty of Mathematics, Informatics and Mechanics, University of Warsaw, Poland}

\address[rouen]{Normandie Universit\'e, LITIS EA4108, NormaStic CNRS FR 3638, IRIB, Universit\'e de Rouen, 76821 Mont-Saint-Aignan Cedex, France}

\fntext[Gabriele]{Supported by the  Italian Ministry of Education (MIUR) project PRIN 2010/2011  ``Automi e Linguaggi Formali: Aspetti Matematici e Applicativi''.}

\fntext[Iuventus]{Supported by the Polish Ministry of Science and Higher Education under the `Iuventus Plus' program in 2015-2016 grant no 0392/IP3/2015/73.
The author is also supported by Polish budget funds for science in 2013-2017 as a research project under the `Diamond Grant' program.}

\begin{abstract}
Given a word $w$ and a Parikh vector $\PV$, an abelian run of period $\PV$ in $w$ is a maximal occurrence of a substring of $w$ having abelian period $\PV$.
Our main result is an online algorithm that, given a word $w$ of length $n$ over an alphabet of cardinality $\sigma$ and a Parikh vector $\PV$, returns all the abelian runs of period $\PV$ in $w$
in time $O(n)$ and space $O(\sigma+p)$, where $p$ is the norm of $\PV$, i.e., the sum of its components.
We also present an online algorithm that computes all the abelian runs with periods 
of norm $p$ in $w$ in time $O(np)$, for any given norm $p$.
Finally, we give an $O(n^2)$-time offline
 randomized algorithm for computing all
 the abelian runs of $w$. Its deterministic counterpart runs in $O(n^2\log\sigma)$ time.
\end{abstract}


\end{frontmatter}


\section{Introduction}

Computing maximal (non-extendable) repetitions in a word is a classical topic in the area of string algorithms (see for example \cite{Smy2013} and references therein). Maximal repetitions of substrings, also called \emph{runs}, give information on the repetitive regions of a word, and are used in many applications, for example in the analysis of genomic sequences.

Kolpakov and Kucherov \cite{KK99} gave the first linear-time algorithm for computing all the runs in a word and conjectured that any word of length $n$ contains less than $n$ runs. Recently, Bannai et al.~\cite{B14,B14bis}, using the notion of Lyndon roots of a run, proved this conjecture and designed a much simpler
algorithm computing the runs.

Here we deal with a generalization of this problem to the commutative setting. 
Recall that an abelian power is a concatenation of two or more words that have the same Parikh vector, i.e., that have the same number of occurrences of each letter of the alphabet. 
For example, $aababa$ is an abelian square, since $aab$ and $aba$ both have two $a$'s and one $b$, i.e., the same Parikh vector $\PV=(2,1)$. 
When an abelian power occurs within a word, one can search for its ``maximal'' occurrence by extending it to the left and to the right character by character without violating the condition on the number of occurrences of each letter. 
Following the approach of Constantinescu and Ilie~\cite{CI2006}, we say that a Parikh vector $\PV$ is an abelian period of a word $w$ if $w$ can be written as $w=u_0u_1 \cdots u_{k-1}u_k$ for some $k\ge 1$ where for $0<i<k$ all the $u_i$'s have the same Parikh vector $\PV$ and the Parikh vectors of $u_0$ and $u_k$ are contained in $\PV$. 
If $k>2$, we say that the word $w$ is periodic with period $\PV$. 
Note that the factorization above is not necessarily unique. For example, $a\cdot bba\cdot bba \cdot \epsilon$ and $\epsilon \cdot abb\cdot abb \cdot a$ ($\epsilon$ denotes the empty word) are two factorizations of the word $abbabba$ both corresponding to the abelian period $(1,2)$. Moreover, the same word can have different abelian periods.
 
 In this paper we define an \emph{abelian run} of period $\PV$ in a word $w$ as an occurrence of a substring $v$ of $w$ such that $v$ is periodic with abelian period $\PV$ and this occurrence cannot be extended to the left nor to the right by one letter into a substring periodic with period $\PV$.
 
For example, let $w=ababaaa$. Then the prefix $ab\cdot ab\cdot a=w[0\pp 4]$ has abelian period $(1,1)$ but it is not an abelian run since the prefix $a\cdot ba \cdot ba\cdot a=w[0\pp 5]$ also has abelian period $(1,1)$. The latter, on the other hand, is an abelian run of period $(1,1)$ in $w$.
 
 Looking for abelian runs in a word can be useful to detect regions in the word where there is some kind of non-exact repetitiveness, for example regions with several consecutive occurrences of a substring or its reversal.
 
Matsuda et al.~\cite{matsuda2014computing} recently presented an offline algorithm for computing all abelian runs of a word of length $n$ in $O(n^2)$ time.
 Notice that, however, the definition of abelian run in~\cite{matsuda2014computing} is slightly different from the one we consider here. We compare both versions in Section \ref{sect-def}.
Basically, our notion of abelian run is more restrictive than the one
 of~\cite{matsuda2014computing}, for which we use the term ``anchored run''.
 
We first present an online algorithm that, given a word $w$ of length $n$ over an alphabet of cardinality $\sigma$ and a Parikh vector $\PV$, returns all the abelian runs of period $\PV$ in $w$ in time  $O(n)$ and space $O(\sigma+p)$, where $p$ is the norm of $\PV$, that is, the sum of its components.
This algorithm improves upon the one given in~\cite{FLLP15} which runs in time $O(n p)$. 
Next, we give an $O(np)$-time online algorithm for computing all the abelian
 runs with periods of norm $p$ of a word of length $n$, for any given $p$.
 Finally, we present an $O(n^2)$ (resp. $O(n^2\log n)$) -time offline
 randomized (resp. deterministic) algorithm for computing all
 the abelian runs of a word  of length $n$.

The rest of this article is organized as follows.
Sect.~\ref{sect-def} introduces central concepts and fixes the notation.
In Sect.~\ref{sect-prev} we review the results on abelian runs given in \cite{matsuda2014computing}.
Sect.~\ref{sect-new} is devoted to the presentation of our main result: a new solution
 for computing the abelian runs for a given Parikh vector.
In Sect.~\ref{sect-newer} we apply this algorithm in a procedure for computing the abelian runs with periods of a given norm.
 Next, in Sect.~\ref{sect-all}, we design a solution for computing
 all the abelian runs, which builds upon the result recalled in Sect.~\ref{sect-prev}.
Finally, we conclude in Sect.~\ref{sect-conc}.

\section{Definitions and Notation}\label{sect-def}

Let $\Sigma=\{a_{1},a_{2},\ldots,a_{\sigma}\}$ be a finite ordered
 alphabet of cardinality $\sigma$, and let $\Sigma^*$ be the set of finite words
 over $\Sigma$. We assume that the mapping between $a_i$ and $i$ can be evaluated in constant time for $1\le i \le \sigma$.
We let $|\s{w}|$ denote the length of the word $\s{w}$.
Given a word $\s{w}=\s{w}[0\pp n-1]$ of length $n>0$, we write $\s{w}[i]$ for the $(i+1)$-th symbol of $\s{w}$
 and, for $0\le i \le j< n$, we write $\s{w}[i\pp j]$ to denote a fragment of $\s{w}$
 from the $(i+1)$-th symbol to the $(j+1)$-th symbol, both included. This fragment is an occurrence of a substring
 $w[i]\cdots w[j]$.
 For $0\le i \le n$, $\s{w}[i\pp i-1]$ denotes
 the empty fragment.
We let $|\s{w}|_a$ denote the number of occurrences of the symbol
 $a\in\Sigma$ in the word $\s{w}$. 
The Parikh vector of $\s{w}$, denoted by $\PVW$,
 counts the
 occurrences of each letter of $\Sigma$ in $\s{w}$, that is, 
 $\PVW=(|\s{w}|_{a_{1}},\ldots,|\s{w}|_{a_{\sigma}})$.
Notice that two words have the same Parikh vector if and only if
 one word is a permutation (i.e., an anagram) of the other.
Given the Parikh vector $\PVW$ of a word $\s{w}$, we let $\PVW [i]$ denote its
 $i$-th component and $|\PVW|$ its norm, defined as the sum of its components.
Thus, for $\s{w}\in\Sigma^*$ and $1\le i\le\sigma$, we have
 $\PVW [i]=|\s{w}|_{a_i}$ and $|\PVW|=\sum_{i=1}^{\sigma}\PVW[i]=|\s{w}|$.
Finally, given two Parikh vectors $\PV,\QV$, we write $\PV\subseteq \QV$ if
 $\PV[i]\le \QV[i]$
 for every $1\le i\le \sigma$. If additionally $\PV\ne \QV$, we write $\PV\subset \QV$ 
 and say that $\PV$ is contained in $\QV$.
 
\begin{definition}[Abelian period~\cite{CI2006}]
\label{def-ap}
A factorization $\s{w}=\s{u}_0\s{u}_1 \cdots \s{u}_{k-1}\s{u}_{k}$ satisfying
$k\ge 1$,  $\PV_{\s{u}_{1}}=\cdots =\PV_{\s{u}_{k-1}}=\PV$, and $\PV_{\s{u}_{0}}\subset \PV \supset \PV_{\s{u}_{k}}$
is called a \emph{periodic factorization} of $w$ with respect to $\PV$.
If a word $\s{w}$ admits such a factorization, we say that $\PV$ is an \emph{(abelian) period} of $\s{w}$.
\end{definition}

We call fragments $\s{u}_0$ and $\s{u}_k$ respectively the \emph{head} and the
 \emph{tail} of the factorization, while the remaining factors are called \emph{cores}.
 Note that the head and the tail are of length strictly smaller than $|\PV|$; in particular they can be empty.
 
Observe that a periodic factorization with respect to a fixed period is not unique. However, it suffices to specify
$|\s{u}_0|$ to indicate a particular factorization; see~\cite{CI2006}. 
Dealing with factorizations of fragments
of a fixed text, it is more convenient to use a different quantity for this aim.
Suppose $w[i\pp j]=u_0\cdots u_k$ is a factorization with respect to abelian period $\PV$ with $p=|\PV|$.
Observe that consecutive starting positions $i_1,\ldots,i_k$ of factors $u_1,\ldots,u_k$ differ by exactly $p$. 
Hence, they share a common remainder modulo $p$, which we call the \emph{anchor} of the factorization.
Note that the anchor does not change if we trim a factorization of $w[i\pp j]$ to a factorization of a shorter fragment, 
or if we extend it to a factorization of a longer fragment.

\begin{definition}[Anchored period]
A fragment $w[i\pp j]$ has an abelian period $\PV$ \emph{anchored} at $k$ if
it has a periodic factorization with respect to $\PV$ whose anchor is $k \bmod p$.
\end{definition}

 If $w$ has a factorization with at least two cores, we say that $w$ is \emph{periodic} with period $\PV$
 (anchored at $k$ if $k \bmod p$ is the anchor of the factorization).

\begin{definition}[Abelian run]
A fragment $w[i\pp j]$ is called an \emph{abelian run} with period $\PV$ if it is periodic with period $\PV$
and maximal with respect to this property (i.e., each of $w[i-1\pp j]$ and $w[i\pp j+1]$ either does not exist
or it is not periodic with period $\PV$).
\end{definition}
We shall often represent an abelian run $w[i\pp j]$ as a tuple $(i,h,t,j)$ where $h$ and $t$ are respectively the lengths 
of the head and the of the tail of a periodic factorization of $w[i\pp j]$ with period $\PV$ 
and at least two cores (see \figurename~\ref{figu-def}). Note that $(i+h)\bmod p$ is the anchor of the factorization, and that $\frac1p(j-i-h-t-1)$
is the number of cores, in particular it is an integer.

Observe that an abelian run with period $\PV$ may have several valid factorizations.
For example, $a\cdot ba \cdot ba\cdot \epsilon$
and $\epsilon \cdot ab\cdot ab \cdot a$ are factorization of a run $w[0..4]$ with period $\PV=(1,1)$ in $w=ababa$.
Therefore the run can be represented as $(0,1,0,4)$ and as $(0,0,1,4)$. 
However, in $v=abab$ only $(0,0,0,3)$ is a representation of $v[0..3]$ as an abelian run with period $\PV=(1,1)$. 
This is because $(0,1,1,3)$ corresponds to a factorization $v[0..3]=a\cdot ba \cdot b$ with one core only,
and such a factorization does not indicate that $v[0..3]$ is an abelian run.

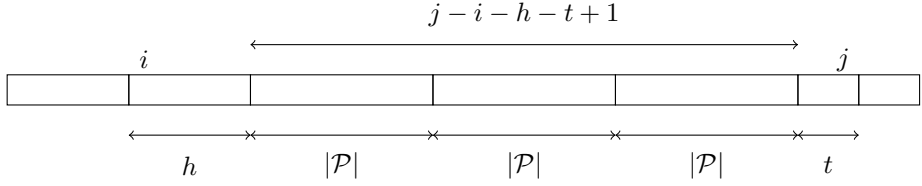
\begin{figure}
\begin{center}
\begin{tikzpicture}[scale=0.8]

\draw (0,0) rectangle (2,.5);
\draw (2,0) rectangle (4,.5);
\draw (4,0) rectangle (7,.5);
\draw (7,0) rectangle (10,.5);
\draw (10,0) rectangle (13,.5);
\draw (13,0) rectangle (14,.5);
\draw (14,0) rectangle (15,.5);

\node at (2.25,.75) {$i$};
\node at (13.75,.75) {$j$};

\draw[<->] (2,-.5) -- (4,-.5) ;
\node at (3,-1) {$h$};

\draw[<->] (13,-.5) -- (14,-.5) ;
\node at (13.5,-1) {$t$};

\draw[<->] (4,-.5) -- (7,-.5) ;
\node at (5.5,-1) {$|\PV|$};

\draw[<->] (7,-.5) -- (10,-.5) ;
\node at (8.5,-1) {$|\PV|$};

\draw[<->] (10,-.5) -- (13,-.5) ;
\node at (11.5,-1) {$|\PV|$};

\draw[<->] (4,1) -- (13,1) ;
\node at (8.5,1.5) {$j-i-h-t+1$};

\end{tikzpicture} 

\caption{\label{figu-def}
The tuple $(i,h,t,j)$ denotes an occurrence of a substring starting at position $i$, ending at position $j$, and having abelian period $\PV$ with head length $h$ and tail length $t$.
}
\end{center}
\end{figure}

Matsuda et al.~\cite{matsuda2014computing} gave a different definition of abelian runs, where maximality is with respect to extending
a fixed factorization. In this paper, we call such fragments anchored (abelian) runs.

\begin{definition}[Anchored run~\cite{matsuda2014computing}]\label{def-anchored}
A fragment $w[i\pp j]$ is a \emph{$k$-anchored abelian run} with period $\PV$ if $w[i\pp j]$ is periodic with period $\PV$ anchored at $k$
and maximal with respect to this property (i.e., each of $w[i-1\pp j]$ and $w[i\pp j+1]$ either does not exist
or it is not periodic with period $\PV$ anchored at $k$).
\end{definition}

Note that every abelian run is an anchored run with the same period (for some anchor). 
The converse is not true, since it might be possible to extend an anchored run preserving the period but not the anchor. 
For example, in the word $w=ababaaa$ considered in the introduction, the fragment $w[0\pp 4]=\epsilon \cdot ab \cdot ab \cdot a$ is a $0$-anchored run
but not an abelian run, since $w[0\pp 5]=a \cdot ba \cdot ba \cdot a$ is periodic with abelian period $(1,1)$.

 Since a factorization
is uniquely determined by the anchor, standard inclusion-maximality is equivalent 
to the condition in Definition~\ref{def-anchored}.

\begin{observation}
Let $w[i\pp j]$ and $w[i'\pp j']$ be fragments of $w$ with abelian period $\PV$ anchored at $k$.
If $w[i\pp j]$ is properly contained in $w[i'\pp j']$ (i.e, $i'< i$ and $j\le j'$, or $i'\le i$ and $j<j'$), then $w[i\pp j]$ is not a $k$-anchored abelian run with period~$\PV$.
\end{observation}

Abelian runs enjoy the same property, but its proof is no longer trivial.
\begin{lemma}\label{lemma-maximality}
Let $w[i\pp j]$ and $w[i'\pp j']$ be fragments of $w$ with abelian period $\PV$.
If $w[i\pp j]$ is properly contained in $w[i'\pp j']$, then $w[i\pp j]$ is not an abelian run with period $\PV$.
\end{lemma}
\begin{proof}
We assume that $i'<i$. The case of $j<j'$ is symmetric. 
For a proof by contradiction suppose that $w[i\pp j]$ is an abelian run and let $w[i\pp j]=u_0\cdots u_k$ be a periodic factorization with period $\PV$
and at least two cores (i.e., satisfying $k\ge 3$).
A periodic factorization of $w[i'\pp j']$ can be trimmed to a factorization $w[i-1\pp j]={v}_0\cdots {v}_{\ell}$.
However, since $w[i\pp j]$ is an abelian run, this factorization must have at most one core (i.e., $\ell\le 2$).
Moreover, $u_0\cdots u_k$ cannot be extended to a factorization of $w[i-1\pp j]=u'_0u_1\cdots u_k$. In other words $u'_0$,
the extension of $u_0$ by one letter to the left, must satisfy $\PV_{u'_0}\not\subseteq \PV$.

Let $p=|\PV|$. The conditions on the number of cores imply $|u_0|+2p\le |w[i\pp j]|$ and $|w[i-1]\pp j]| < |v_0|+2p$.
Consequently, $|u'_0|=|u_0|+1< |v_0|$, i.e., $u'_0$ is a proper prefix of $v_0$. This yields
$\PV_{u'_0}\subset \PV_{v_0} \subset \PV$, which is in contradiction with $\PV_{u'_0}\not\subseteq \PV$.
\qed\end{proof}

\begin{corollary}\label{coro-run}
Let $w$ be a word. For a fixed Parikh vector $\PV$, there is at most one abelian run with abelian period $\PV$ starting at each position of $w$.
\end{corollary}

\section{Previous Work}\label{sect-prev}

Matsuda et al.~\cite{matsuda2014computing} presented an algorithm that computes all the anchored runs
 of a word $w$ of length $n$ in $O(n^2)$ time and space complexity.
The initial step of the algorithm is to compute maximal abelian powers in $w$.
Recall that an abelian power is a concatenation of several abelian-equivalent words. 
In other words, an abelian power of period $\PV$ is a word admitting a periodic factorization with respect to $\PV$ with an empty head, an empty tail
and at least two cores. A fragment $w[i\pp j]$ is a maximal abelian power if
it cannot be extended to a longer power of period  $\PV$ (preserving the anchor). 
Formally, the maximality conditions~are
\begin{enumerate}
\item $\PV_{\s{w}[i-p\pp i-1]} \ne \PV_{\s{w}[i\pp i+p-1]}$ or $i - p < 0$, and
\item $\PV_{\s{w}[j-p+1\pp j]} \ne \PV_{\s{w}[j+1\pp j+p]}$ or $j + p \ge n$,
\end{enumerate}
where $p=|\PV|$.

The approach of~\cite{matsuda2014computing} is to first compute all the abelian squares using the algorithm by Cummings \& Smyth~\cite{Cummings_weakrepetitions}. 
The next step is to group squares into maximal abelian powers. For this, it suffices to merge pairs of overlapping
abelian squares of the form $w[i\pp i+2p-1]$ and $w[i+p\pp i+3p-1]$. This way maximal abelian powers are computed in $O(n^2)$ time.

Observe that there is a natural one-to-one correspondence between maximal abelian powers and anchored runs: it suffices to trim
the head and the tail of the factorization of an anchored run to obtain a maximal abelian power. 
Hence, the last step of the algorithm is to compute the maximal head and tail by which each abelian power can be extended.
This could be done naively in $O(n^3)$ time overall,
but a clever computation enables to
 find all the abelian runs in time and space $O(n^2)$
 (see~\cite{matsuda2014computing} for further details).
 
 In Section~\ref{sect-all}, we extend this result to compute the abelian runs only rather than all the anchored runs.
 Both these algorithms work offline: they need to know the whole word before reporting any abelian run.
In the following two sections we give several online algorithms, which are  able to report a run ending at position $i-1$ of a word $w$
before reading $w[i+1]$ and the following letters. Clearly, not knowing $w[i]$ one cannot decide whether the run could be extended to the right,
so this is the optimal delay. However, these methods are restricted to finding runs of a given period
 or a given norm of the periods, respectively.

\section{Computing Abelian Runs with Fixed Parikh Vector}\label{sect-new}

In this section we present our online solution for computing all the abelian
 runs of a given Parikh vector $\PV$ of norm $p$ in a given word $w$.
 The algorithm works in $O(n)$ time and $O(\sigma+p)$ space where $n=|w|$.
 
 First, in Sect.~\ref{sub-anchored}, we show how to compute all anchored runs of period~$\PV$.
 Later, in Sect.~\ref{sub-abelian}, we modify the algorithm to return abelian runs only.
 We conclude in Sect.~\ref{sub-example} with an example course of actions in our solution.

\subsection{Algorithm for Anchored Runs}\label{sub-anchored}
We begin with a description of data maintained while scanning the string $w$. 
For an integer $k$, let $B_{i}[k]$ be the starting position of the longest suffix of $w[0\pp i]$ which 
has period $\PV$ anchored at $k$. If there is no such a suffix, we set $B_{i}[k] =\infty$.
Since this notion depends on $k\bmod p$ only, we store $B_i[k]$ for $0\le k < p$ only.

Let $b_i$ be the starting position of the longest suffix of $w[0\pp i]$
whose Parikh vector is contained in or equal to $\PV$.
In other words, we have $\PV_{w[b_i\pp i]}\subseteq \PV$ and $\PV_{w[b_i-1\pp i]}\not\subseteq \PV$ (or $b_i=0$).
Note that $b_i=i+1$ if $w[i]$ does not occur in $\PV$.

Observe that the tail of any periodic factorization of a suffix of $w[0\pp i]$ must be contained in $w[b_i\pp i]$.
This leads to the following characterization:
\begin{lemma}\label{lemma-char-simple}
Let $0\le i < |w|$. We have $B_i[k]\le k$ for $b_i\le k \le i+1$ and $B_i[k]=\infty$
for $i-p+1<k<b_i$.
\end{lemma}
\begin{proof}
For $b_i\le k \le i+1$, the fragment $w[k\pp i]$ has abelian period $\PV$ anchored at $k$. (The underlying factorization
has empty head, no cores and tail $w[k\pp i]$, unless $k=b_i=i-p+1$, when the factorization has one core, empty head and empty tail).
Hence, we have $B_i[k]\le k$ directly from the definition.

For $i-p+1<k<b_i$, the tail of the factorization with anchor $k\bmod p$ would need to start at position $k$,
which is impossible (see \figurename~\ref{figu-case1}).
\qed\end{proof}

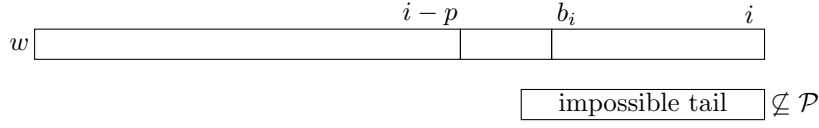
\begin{figure}[t]
\begin{center}
\begin{tikzpicture}[scale=0.8]

\draw (.5,0) rectangle (7.5,.5);
\draw (7.5,0) rectangle (9,.5);
\draw (9,0) rectangle (12.5,.5);

\node at (.25,.25) {$w$};
\node at (7,.75) {$i-p$};
\node at (9.25,.75) {$b_i$};
\node at (12.25,.75) {$i$};

\draw (8.5,-.5) rectangle (12.5,-1);

\node at (13,-.75) {$\not\subseteq \PV$};
\node at (10.5,-.75) {impossible tail};

\end{tikzpicture} 
\caption{\label{figu-case1}
$B_i[k]=\infty$ for $i-p+1<k<b_i$.
}
\end{center}
\end{figure}

The values $b_{i-1}$ and $b_i$ are actually sufficient to describe $B_i$ based on $B_{i-1}$.

\begin{lemma}\label{lemma-algo}
For $0\le i < |w|$ the following equalities hold:
\begin{enumerate}
  \item\label{case1} $B_i[k]=\infty \ne B_{i-1}[k]$ for $\max(i-p+1,b_{i-1}) \le k  < b_i$,
  \item\label{case2} $B_i[k]=B_{i-1}[k]$ for $b_i \le k \le  i$ and for $i-p+1 < k  < b_{i-1}$,
    \item\label{case3} $B_i[i+1] = b_i$ if $b_i > i-p+1$ and $B_i[i+1]=B_{i-1}[i-p+1]$ otherwise.
\end{enumerate}
\end{lemma}

\begin{proof}
Lemma~\ref{lemma-char-simple} implies that $B_{i}[k]=\infty$ for $i-p+1<k<b_i$ and $B_{i-1}[k]=\infty$ for $i-p+1<k<b_{i-1}$
(hence $B_{i-1}[k]=B_{i}[k]$ in this latter case).
For $b_{i}\le k  < i$, we have $\PV_{w[k \pp i]}\subseteq \PV$, so
we can extend the factorization of a suffix of $w[0\pp i-1]$ whose tail starts at position $k$ (see \figurename~\ref{figu-case2}).

Finally, note that $B_i[i+1]$ is the starting position of the maximal suffix of $w[0\pp i]$
with an empty-tail periodic factorization. If $\PV_{w[i-p+1\pp i]}\ne \PV$ (i.e., if $b_i> i-p+1$), this is just $w[b_i\pp i]$.
Otherwise, we can extend the factorization of a suffix of $w[0\pp i-1]$ whose tail starts at position $i-p+1$.
  \qed
\end{proof}

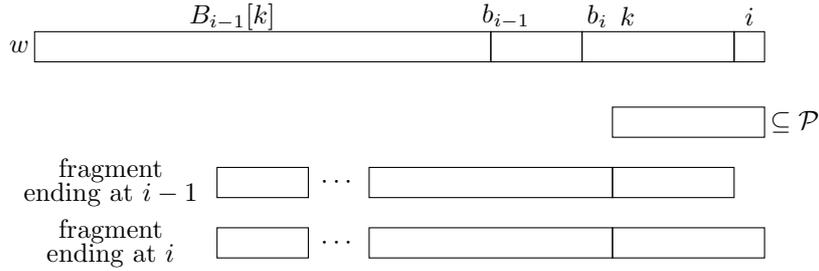
\begin{figure}[b]
\begin{center}
\begin{tikzpicture}[scale=0.8]

\draw (.5,0) rectangle (8,.5);
\draw (8,0) rectangle (9.5,.5);
\draw (9.5,0) rectangle (12,.5);
\draw (12,0) rectangle (12.5,.5);

\node at (.25,.25) {$w$};
\node at (3.75,.75) {$B_{i-1}[k]$};
\node at (8.25,.75) {$b_{i-1}$};
\node at (10.25,.75) {$k$};
\node at (9.75,.75) {$b_i$};
\node at (12.25,.75) {$i$};

\begin{scope}[yshift=.75cm]
\draw (10,-1.5) rectangle (12.5,-2);
\node at (13,-1.75) {$\subseteq \PV$};

\draw (3.5,-2.5) rectangle (5,-3);
\draw (6,-2.5) rectangle (10,-3);
\draw (10,-2.5) rectangle (12,-3);
\node at (5.5,-2.75) {$\cdots$};

\draw (3.5,-3.5) rectangle (5,-4);
\draw (6,-3.5) rectangle (10,-4);
\draw (10,-3.5) rectangle (12.5,-4);
\node at (5.5,-3.75) {$\cdots$};

\node at (1.75,-2.55) {fragment};
\node at (1.75,-2.95) {ending at $i-1$};
\node at (1.75,-3.55) {fragment};
\node at (1.75,-3.95) {ending at $i$};
\end{scope}

\end{tikzpicture} 
\caption{\label{figu-case2}
$B_i[k]=B_{i-1}[k]$ for $b_i \le k \le i$.
}
\end{center}
\end{figure}

Having read letter $w[i]$, we need to report anchored runs which end at position $i-1$.
For this, we use the following characterization.
\begin{lemma}\label{lem-report-anchored}
Let $i-p<k\le i$. A fragment $w[b\pp i-1]$ is a $k$-anchored run with period $\PV$
if and only if $B_{i-1}[k] = b\le k-2p$ and $B_i[k]>b$.
\end{lemma}
\begin{proof}
Clearly an anchored run ending at position $i-1$ must be a left-maximal suffix of $w[0\pp i-1]$
with a given anchor. Moreover, we must have $b\le k-2p$ so that the factorization has at least two cores
and $B_i[k]>b$ due to right-maximality.
It is easy to see that these conditions are sufficient.
\qed\end{proof}

By Lemma~\ref{lemma-algo}, most entries of $B_i$ are inherited from $B_{i-1}$,
so we use a single array $B$ and having read $w[i]$, we update its entries.
As evident from Lemma~\ref{lem-report-anchored}, each anchored run to be reported corresponds
to a modified entry.

The algorithm \CALL{AnchoredRun}{\PV,p,w,n} in \figurename~\ref{algo-run0} implements our approach.
The \textbf{while} loop increments $k$ from $b_{i-1}$ to $b_i$. 
For $k>i-p$, we set $B[k]$ to $\infty$ and possibly report a run. 
Note that $k=i-p+1$ is within the scope of Case~\ref{case3} rather than Case~\ref{case1} in Lemma~\ref{lemma-algo}. 
However, later we set $B[i+1]$ to $b_{i}$ if $b_{i}>i-p+1$ (as described in Case~\ref{case3}).
Nevertheless, if an $(i+1)$-anchored run needs to be reported, we have $B_{i-1}[i-p+1]<\infty=B_{i}[i-p+1]$,
so $b_{i-1}\le i-p+1$ and thus $k=i-p+1$ is considered in the loop.

  \begin{figure}[b!]
\begin{center}
\begin{algo}[rules]{AnchoredRun}{\PV,p,w,n}
  \SET{k}{0}
  \SET{B[0]}{k}
  \label{loop-for-begin}\DOFORI{i}{0}{n}
    \label{loop-while-begin}\DOWHILE{k \le n \textbf{\mbox{ and }} (i = n \textbf{\mbox{ or }} \PV_{\s{w}[k\pp i]} \not\subseteq \PV)}
      \IF{k > i-p}
        \SET{b}{B[k\bmod p]}
        \SET{B[k\bmod p]}{\infty}
        \IF{b\le k-2p}
        	\SET{h}{(k-b) \bmod p}
        	\SET{t}{i-k}
           \CALL{Output}{b, h,t, i-1}
        \FI
      \FI
      \label{loop-while-end}\SET{k}{k+1}
    \OD
     \IF{k > i-p+1}
      \label{loop-for-end}\SET{B[(i+1) \bmod p]}{k}
    \FI
  \OD
\end{algo}
\caption{\label{algo-run0}
Algorithm computing all the anchored runs of period $\PV$ of norm $p$ in
 a word $w$ of length $n$.
}
\end{center}
\end{figure}

\begin{theorem}
The algorithm \CALL{AnchoredRun}{\PV,p,w,n} computes all the anchored runs with period $\PV$ of norm $p$ in a word
 $w$ of length $n$ in time $O(n)$ and additional space $O(\sigma+p)$.
\end{theorem}
\begin{proof}
The correctness of the algorithm comes from Lemmas~\ref{lemma-algo}-\ref{lem-report-anchored} and the discussion above.
The external \textbf{for} loop in lines~\ref{loop-for-begin}--\ref{loop-for-end} runs  $n+1$ times.
The internal \textbf{while} loop in lines~\ref{loop-while-begin}--\ref{loop-while-end}
 cannot iterate more than $n+1$ times since it starts with $k$ equal to $0$
 and ends when $k$ is equal to $n$ and $k$ can only be incremented by $1$ (in line~\ref{loop-while-end}).
The test $\PV_{\s{w}[k\pp i]} \not\subseteq \PV$ in line~\ref{loop-while-begin}
 can be realized in constant time once we store $\PV_{\s{w}[k\pp i]}$ and a counter of its components
 for which the value is greater than in $\PV$. 
This data needs to be updated once we increment $i$ in the \textbf{for} loop and $k$
 in line~\ref{loop-while-end}. 
We then need to increment the component $w[i]$ or decrement the component $w[k]$ of $\PV_{\s{w}[k\pp i]}$, respectively. 
The global counter needs to be updated accordingly.
All the other operations run in constant time. 
Thus the total time complexity of the algorithm is $O(n)$.
The space complexity comes from the number of counters ($\sigma$) and the size of the array $B$ ($p$).
\qed\end{proof}

Note that the space consumption can be reduced to $O(p)$ at the price of introducing (Monte Carlo) randomization.
Instead of storing the Parikh vectors in a plain form, we can use dynamic hash tables~\cite{hashing} so that the size is proportional
to the number of non-zero entries.

\subsection{Algorithm for Abelian Runs}\label{sub-abelian}

  \begin{figure}[b!]
 \begin{center}
\begin{algo}[rules]{Run}{\PV,p,w,n}
  \SET{k}{0}
  \SET{L}{\emptyset}
  \SET{B[0]}{k}
  \SET{Ptr[0]}{\CALL{insertAtTheEnd}{L,0}}
  \DOFORI{i}{0}{n}
    \SET{b_{\min}}{B[\CALL{getFirst}{L}]}
    \DOWHILE{k \le n \textbf{\mbox{ and }} (i = n \textbf{\mbox{ or }} \PV_{\s{w}[k\pp i]} \not\subseteq \PV)}
      \IF{k > i-p}
        \SET{b}{B[k\bmod p]}
        \SET{B[k\bmod p]}{\infty}
        \CALL{Delete}{L,Ptr[k\bmod p]}
        \SET{Ptr[k\bmod p]}{\Nil}
        \IF{b=b_{\min} \textbf{\mbox{ and }} B[\CALL{getFirst}{L}] > b \textbf{\mbox{ and }} b\le k-2p}
            \SET{h}{(k-b)\bmod p}
            \SET{t}{i-k}
           \CALL{Output}{b, h, t, i}
        \FI
      \FI
      \SET{k}{k+1}
       \IF{k > i- p+1}
      \SET{B[(i+1) \bmod p]}{k}
      \SET{Ptr[(i+1) \bmod p]}{\CALL{insertAtTheEnd}{L,(i+1) \bmod p}}
    \FI
    \OD
  \OD
\end{algo}
\caption{\label{algo-run}
Algorithm computing all the abelian runs of period $\PV$ of norm $p$ in
 a word $w$ of length $n$.
}
\end{center}
\end{figure}

In this section we extend our algorithm so that it reports abelian runs only.
For an offline solution, we could simply determine the anchored runs (using the procedure developed above) and filter out those which are not maximal. 
However, in order to obtain an online algorithm, we need a more subtle approach,
which is based on the following characterization. 

\begin{lemma}\label{lem-report-abelian}
A fragment $w[b\pp i-1]$ is an abelian run with period $\PV$
if and only if it is an anchored run (with period $\PV$)
and for each $k'$ the inequalities $B_{i-1}[k']\ge b$ and $B_{i}[k']>b$ hold.
\end{lemma}

\begin{proof}
By Lemma~\ref{lemma-maximality}, an abelian run of period $\PV$ cannot be properly contained in a fragment with period $\PV$
(anchored at some $k'$).
Conditions involving $B_{i-1}[k']$ and $B_i[k']$ enforce left-maximality and right-maximality, respectively.
Since each abelian run is an anchored run (with the same period) and since all
anchored runs are periodic, the claim follows.
\qed\end{proof}
To apply Lemma~\ref{lem-report-abelian}, it suffices to find an anchor $k$ such that $b=B_{i-1}[k]=\min_{k'}B_{i-1}[k']<\min_{k'}B_{i}[k']$.
There can be several such anchors and in case of ties we are going to detect the one for which the factorization of $w[b\pp i-1]$ has shortest tail.
This factorization maximizes the number of cores, so if $w[b\pp i-1]$ is an anchored run with any anchor, it is with that one in particular.
Note that the \textbf{while} loop in lines~\ref{loop-while-begin}--\ref{loop-while-end} of Algorithm \textsc{AnchoredRun} processes anchors $b_{i-1}\le k < b_i$ in the order of decreasing tail lengths and updates the underlying values $B[k]$ from $B_{i-1}[k]$ to $B_{i}[k]$.
For the sought anchor $k$ this update strictly increases the value $\min_{k'}B[k']$, and moreover this is 
the first increase of the minimum within a given iteration of the outer \textbf{for} loop.
Hence, we record the original minimum and check for an abelian run only if $\min_{k'}B[k']$ increases from that value.

To implement the procedure described above, we need to efficiently compute  the smallest element in the array $B$. 
For this, recall that $B[j]$ can only be modified from $\infty$ to $k$ 
(and the value $k$ does not decrease throughout the algorithm) or from some value  back to $\infty$.
We maintain a doubly-linked list $L$ of all indices $j$ with finite $B[j]$ 
such that the order of indices $j$ in the list is consistent with the order of values $B[j]$.
To update the list, it suffices to insert the index $j$ to the end of list while
 setting $B[j]$ to $k$, and remove it from the list setting $B[j]$ to $\infty$.
Then the smallest value in $B$ is attained at an argument stored as the first element of the
 list $L$ (or $\infty$, if the list is empty).

The algorithm \CALL{Run}{\PV,p,w,n}, depicted in \figurename~\ref{algo-run},
 implements the approach described above.
It uses the following constant-time functions to operate on lists:
\begin{itemize}
\item
 \CALL{insertAtTheEnd}{L,e} that inserts $e$ at the end of
 the doubly-linked list $L$ and returns a pointer to the location of $e$ in the list;
\item
\CALL{Delete}{L,\ptr} that deletes the element pointed by $\ptr$ from the
doubly-linked list $L$;
\item
\CALL{getFirst}{L} that returns the first element of the list $L$ (0 if the list is empty).
\end{itemize}
The algorithm also uses an array $Ptr$ which maps any anchor $j$ to a pointer of the corresponding location in the list $L$ 
(or $\Nil$ if $B[j]=\infty$). 

The discussion above proves that \CALL{Run}{\PV,p,w,n} correctly computes abelian runs with period $\PV$ in $w$. 
Its running time is the same as that of \CALL{AnchoredRun}{\PV,p,w,n} since the structure of the computations remains
the same while additional instructions run in constant time. 
Memory consumption is still $O(p+\sigma)$ because both $L$ and $Ptr$ take $O(p)$ space.

\begin{theorem}
The algorithm \CALL{Run}{\PV,p,w,n} computes all the abelian runs with period $\PV$ of norm $p$ in a word
 $w$ of length $n$ in time $O(n)$ and additional space $O(\sigma+p)$, which can be reduced to $O(p)$ using randomization.
\end{theorem}

\subsection{Example}\label{sub-example}

Let us see the behaviour of the algorithm on
$\Sigma = \{\sa{a},\sa{b}\}$, $w=\sa{abaababaabbb}$ and $\PV=(2,2)$:\\

$k=0 \quad B = [0, \infty, \infty, \infty] \quad L= (0)$

$i=0 \quad  \PV_{w[0]}\subseteq \PV \quad B = [0, 0, \infty, \infty] \quad L= (0,1)$

$i=1 \quad \PV_{w[0\pp 1]}\subseteq \PV  \quad B = [0, 0, 0, \infty] \quad L= (0,1,2)$

$i=2 \quad \PV_{w[0\pp 2]}\subseteq \PV  \quad B = [0, 0, 0, 0] \quad L= (0,1,2,3)$

$i=3 \quad \PV_{w[0\pp 3]}\not\subseteq \PV \quad b=0 \quad B=[\infty, 0,0,0] \quad L= (1,2,3) \quad k=1$ 

\hspace{1.1cm}$\PV_{w[1\pp 3]}\subseteq \PV \quad  B = [1, 0, 0, 0] \quad L= (1,2,3,0)$

$i=4  \quad \PV_{w[1\pp 4]}\subseteq \PV$

$i=5 \quad  \PV_{w[1\pp 5]}\not\subseteq \PV \quad k=2$

\hspace{1.1cm}$\PV_{w[2\pp 5]}\not\subseteq \PV \quad b=0 \quad B = [1, 0, \infty, 0] \quad L= (1,3,0) \quad k=3$

\hspace{1.1cm}$\PV_{w[3\pp 5]}\subseteq \PV \quad  B = [1, 0, 3, 0] \quad L= (1,3,0,2)$

$i=6 \quad \PV_{w[3\pp 6]}\subseteq \PV$

$i=7 \quad  \PV_{w[3\pp 7]}\not\subseteq \PV \quad k=4$

\hspace{1.1cm}$\PV_{w[4\pp 7]}\subseteq \PV$

$i=8 \quad  \PV_{w[4\pp 8]}\not\subseteq \PV \quad k=5$

\hspace{1.1cm}$\PV_{w[5\pp 8]}\not\subseteq \PV \quad b=0 \quad  B = [1, \infty, 3, 0] \quad L= (3,0,2)\quad k=6$

\hspace{1.1cm}$\PV_{w[6\pp 8]}\subseteq \PV \quad  B = [1, 6, 3, 0] \quad L= (3,0,2,1)$

$i=9  \quad \PV_{w[6\pp 9]}\subseteq \PV$

$i=10 \quad  \PV_{w[6\pp 10]}\not\subseteq \PV \quad k=7$

\hspace{1.25cm}$\PV_{w[7\pp 10]}\subseteq \PV$

$i=11 \quad  \PV_{w[7\pp 11]}\not\subseteq \PV \quad k=8$

\hspace{1.25cm}$\PV_{w[8\pp 11]}\not\subseteq \PV \quad b=1 \quad B=[\infty,6,3,0] \quad L=(3,2,1) \quad k=9$

\hspace{1.25cm}$\PV_{w[9\pp 11]}\not\subseteq \PV \quad b=6 \quad B=[\infty,\infty,3,0] \quad L=(3,2)\quad k=10$

\hspace{1.25 cm}$\PV_{w[10\pp 11]}\subseteq \PV \quad B=[10,\infty,3,0] \quad L=(3,2,0)$

$i=12$

\hspace{1.25cm}$b=3 \quad B = [10, \infty, \infty, 0] \quad L= (3,0) \quad k = 11$

\hspace{1.25cm}$b =0 \quad B = [10, \infty, \infty, \infty] \quad L= (0)$

\hspace{1.25cm}$h=3 \quad t=1 \quad \CALL{Output}{0,3,1,11} \quad k=12$

\hspace{1.25cm}$b=10 \quad B=[\infty,\infty,\infty,\infty] \quad L=() \quad k=13$


\section{Computing Abelian Runs with Fixed Parikh Vector Norm}\label{sect-newer}
In this section we develop an $O(np)$-time algorithm to compute all abelian runs
with periods of norm $p$. First, we describe the algorithm for anchored runs
and later generalize it to abelian runs.
\subsection{Anchored Runs}
Let us start with a simple offline algorithm which works in $O(n)$ time to compute $k$-anchored runs with period of norm $p$ for fixed values $p$ and $k$.
This method is similar to the algorithm of Matsuda et al.~\cite{matsuda2014computing} briefly described in Section~\ref{sect-prev}.
Namely, it suffices to compute maximal abelian powers with periods of norm $p$ anchored at $k$, and then extend them by a head and a tail.

Define a \emph{block} as any fragment of the form $w[i\pp i+p-1]$ such that $i \equiv k \pmod p$.
Note that the cores in decompositions  with anchor $k\bmod p$ are blocks. 
Finding $k$-anchored powers with periods of a given norm $p$ is very easy if the anchor is fixed.
We consider consecutive blocks, naively check if they are abelian-equivalent and merge any
maximal chains of abelian-equivalent blocks. Determining the head and the tail of the $k$-anchored runs is also simple.
For each $i \equiv k \pmod p$ we compute the longest suffix of $w[0\pp i-1]$ and the longest prefix of $w[i+p\pp n-1]$
whose Parikh vectors are contained in $\PV_{w[i\pp i+p-1]}$.

This approach can be implemented online in $O(\sigma+p)$ space as follows: we scan consecutive blocks and (naively)
check their abelian equivalence. Whenever we read a full block (say, starting at position $i$),
we compute the longest suffix $w[b_{i-1}\pp i-1]$ of $w[0\pp i-1]$ whose Parikh vector is contained in $\PV_{w[i\pp i+p-1]}$.
This gives a periodic factorization of $w[b_{i-1}\pp i+p-1]$ anchored at $k \bmod p$. 
We then try to extend it to the right while reading further characters. 
Once it is impossible to extend the factorization, say by letter $w[j+1]$, we declare $w[b_{i-1}\pp j]$ as a maximal fragment with period $\PV_{w[i\pp i+p-1]}$ anchored at $k$.
If the decomposition has at least two cores, we report an anchored run. 
If we succeed to extend by a full block (i.e., if $\PV_{w[i\pp i+p-1]}=\PV_{w[i+p\pp i+2p-1]}$), we do not restart the algorithm but instead we continue to extend the factorization.
This way, we guarantee that $b_{i-1}>i-p$ whenever we start building a new factorization.

Clearly, the procedure described above computes all $k$-anchored runs with period of norm $p$. 
To compute all anchored runs, we simply run it in parallel for all $p$ possible anchors.

\begin{theorem}
There is an algorithm which computes online all the anchored runs with periods of norm $p$ in a word
 $w$ of length $n$ over an alphabet of size $\sigma$ in time $O(n p)$ and additional space $O(p(\sigma+p))$, which can be reduced to $O(p^2)$ using randomization.
\end{theorem}

\subsection{Abelian Runs}
Let us first slightly modify the algorithm presented in the previous section.
Observe that whenever we start a new phase having just read a block $w[i\pp i+p-1]$, instead of performing the computations using
a simple procedure described above, we could launch the algorithm of Section~\ref{sub-anchored} for $w[\max(i-p,0)\pp ]$ and $\PV=\PV_{w[i\pp i+p-1]}$, simulate it until it needs to read $w[i+p]$ and then feed it with newly read letters until the maximal extension of $w[i\pp i+p-1]$ anchored at $k$ is found
(i.e., until the respective entry of the $B$ array is set to $\infty$). 
Other anchored runs output by the algorithm should be ignored, of course.
As before, if such a process is running while we have completed reading a subsequent block, we do not start a new phase.

It is easy to see that such an algorithm is equivalent to the previous one.
However, if we use the algorithm of Section~\ref{sub-abelian} instead, we automatically get a possibility to
check whether the maximal extension of $w[i\pp i+p-1]$ anchored at $k$ is a maximal fragment with period $\PV_{w[i\pp i+p-1]}$.
Note that we start the simulation at a position $\max(i-p,0)$ which is smaller than $b_{i-1}$, unless the latter is $0$.
This guarantees that left-maximality is correctly verified despite the fact the fragment prior to position $i-p$ is ignored in the simulation. 
As before, we disregard any other abelian run that the algorithm of Section~\ref{sub-abelian} may return. 
We run this process in parallel for all possible anchors to guarantee that each abelian run with period of norm $p$ is reported exactly once. 
More precisely, in ambiguous cases a run is reported for the anchor corresponding to the factorization with shortest tail,
just as in Section~\ref{sub-abelian}.

\begin{theorem}
There is an algorithm which computes online all the abelian runs with periods of norm $p$ in a word
 $w$ of length $n$ over an alphabet of size $\sigma$ in time $O(n p)$ and additional space $O(p(\sigma+p))$, which can be reduced to $O(p^2)$ using randomization.
\end{theorem}

\section{Offline Algorithm for Computing All Abelian Runs}\label{sect-all}

In this section we present an $O(n^2)$-time offline algorithm which computes all the abelian runs.
As a starting point, we use the set of all anchored runs computed by the algorithm by Matsuda et al. (see Section~\ref{sect-prev}).
Recall that all abelian runs are anchored runs with the same period.
Hence, it suffices to filter out those anchored runs which are properly contained in another anchored run with the same period. We also need to make
sure that every abelian run is reported once only (despite possibly being $k$-anchored for different anchors $k$).

Note that this filtering can be performed independently for distinct periods. If we have a list of anchored runs with a fixed period, sorted
by the starting position, it is easy to retrieve the abelian runs of that period with a single scan of the list. 
Ordering by the starting position can be performed together for all periods so that it takes $O(r+n)$ time where $r$ is the number of all anchored runs. 
Hence, the main difficulty is grouping according to the period. 
For this, we shall assign to each fragment of $w$ an \emph{identifier}, so that two fragments are abelian-equivalent if and only if their identifiers are equal. 
The identifiers of periods can be easily retrieved since given a $k$-anchored run, we can easily locate one of the cores of the underlying factorization.

Thus, in the remaining part of this section we design a naming algorithm which assigns the identifiers.
A naive solution would be to generate the Parikh vectors of all substrings of
$w$, sort these vectors removing duplicates, and give each fragment a rank of its Parikh vector in that order.
However, already storing the Parikh vectors can take prohibitive $\Theta(n^2\sigma)$ space.

To overcome this issue, we use the concept of \emph{diff-representation}, originally introduced in the context of abelian periods~\cite{DBLP:conf/stacs/KociumakaRR13}.
Observe that in a sense the Parikh vectors of fragments can be generated efficiently: for a fixed $p$,
we can first generate $\PV_{w[0\pp p-1]}$, then update it to $\PV_{w[1\pp p]}$, and so on until we reach $\PV_{w[n-p\pp n-1]}$. 
In other words, the Parikh vectors of all fragments of length $p$ can be represented in a sequence so that
the total Hamming distance of the adjacent vectors is $O(n)$.
The diff-representation, designed to manipulate sequences satisfying such a property, is formally defined as a sequence
of single-entry changes such that the original sequence of vectors is a subsequence of intermediate results when applying this operations
starting from the null vector (of the fixed dimension $r$).
 Note that the diff-representation of a sequence of Parikh vectors of all fragments of $w$ can be computed in time $O(n^2)$ proportional
to its size. 
 The following result lets us efficiently assign identifiers to its elements.

\begin{lemma}[\cite{DBLP:conf/stacs/KociumakaRR13}]
Given a sequence of vectors of dimension $r$ represented using a diff-representation of size $m$, consider the problem of assigning  
integer identifiers of size $n^{O(1)}$ so that equality of vectors is equivalent to equality of their identifiers.
It can be solved in $O(r+m\log r)$ time using a deterministic algorithm and in $O(r+m)$ time using a Monte Carlo algorithm
which is correct with high probability ($1-\frac{1}{(r+m)^c}$ where $c$ can be chosen arbitrarily large).
\end{lemma}

In our setting this yields the following result.

\begin{theorem}
There exists an $O(n^2)$-time randomized algorithm (Monte Carlo, correct with high-probability) which computes
all abelian runs in a given word of length $n$. Additionally, there exists an $O(n^2\log\sigma)$-time deterministic algorithm solving the same
problem.
\end{theorem}

\section{Conclusions}\label{sect-conc}

We gave algorithms that, given a word $w$ of length $n$ over an alphabet of cardinality $\sigma$,
 return all the abelian runs of a given period $\PV$ in $w$ in time $O(n)$ and space $O(\sigma+p)$, or all the abelian runs with periods of a given norm $p$ in time $O(np)$ and space $O(p(\sigma+p))$.
These algorithms work in an online manner.
We also presented an $O(n^2)$ (resp. $O(n^2\log n)$)-time offline
 randomized (resp. deterministic) algorithm for computing all
 the abelian runs in a word of length~$n$.
One may wonder if it is possible to reduce further the complexities of these latter algorithms.
We believe that further combinatorial results on the structure of the abelian runs in a word could lead to novel solutions.

\end{document}